\documentclass[conference]{IEEEtran}
\ifCLASSINFOpdf
\else
\fi

\usepackage{verbatim}
\usepackage{amsfonts,amsmath,mathrsfs,amssymb,amsbsy,amsthm}
\usepackage[final]{graphicx}
\usepackage{subcaption}
\usepackage[font={footnotesize}]{caption}
\usepackage[numbers,sort&compress]{natbib}

\usepackage{lipsum}

\newtheorem{thm}{Theorem}

\newtheorem{lem}{Lemma}

\newtheorem{defn}{Definition}
\newtheorem{remk}{\textit{Remark}}

\theoremstyle{definition}
\newtheorem{ex}{\textbf{Example}}

\begin{document}
%
\title{\LARGE{Topological Interference Management with just Retransmission: \\ What are the ``Best'' Topologies?}}

\author{\IEEEauthorblockN{Navid Naderializadeh, Aly El Gamal and A. Salman Avestimehr}
  \IEEEauthorblockA{
    Department of Electrical Engineering,
    University of Southern California\\
    Emails: naderial@usc.edu, aelgamal@usc.edu, avestimehr@ee.usc.edu} }


%


\maketitle

\begin{abstract}
We study the problem of interference management in fast fading wireless networks, in which the transmitters are only aware of network topology. We consider a class of retransmission-based schemes, where transmitters in the network are only allowed to resend their symbols in order to assist with the neutralization of interference at the receivers. We introduce a necessary and sufficient condition on the network topology, under which half symmetric degrees-of-freedom (DoF) is achievable through the considered retransmission-based schemes. This corresponds to the ``best'' topologies since half symmetric DoF is the highest possible value for the symmetric DoF in the presence of interference. We show that when the condition is satisfied, there always exists a set of carefully chosen transmitters in the network, such that by retransmission of their symbols at an appropriate time slot, we can neutralize all the interfering signals at the receivers. Quite surprisingly, we also show that for any given network topology, if we cannot achieve half symmetric DoF by retransmission-based schemes, then there does not exist any linear scheme that can do so. We also consider a practical network scenario that models cell edge users in a heterogeneous network, and show that the characterized condition on the network topology occurs frequently. Furthermore, we numerically evaluate the achievable rates of the DoF-optimal retransmission-based scheme in such network scenario, and show that its throughput gain is not restricted to the asymptotic DoF analysis.
\end{abstract}


%
\IEEEpeerreviewmaketitle
\vspace*{-0.23in}
\section{Introduction}\label{sec:intro}

As wireless networks grow in size and become more and more distributed and heterogeneous, coordination among users and acquisition of channel state information (CSI) gets much more complicated. As a result, there has been a growing interest in development of interference management techniques that rely only on a minimal level of channel state knowledge at the transmitters.

Our focus in this paper is on the scenario in which interference management primarily relies on a coarse knowledge about channel states in the network, namely the  ``topology'' of the network. Network topology simply refers to the 1-bit feedback information for each link between each transmitter and each receiver, indicating whether or not the signal of the transmitter is received above the noise floor at the corresponding receiver. There have been several prior works in the literature that have considered this scenario. In~\cite{localview}, it is assumed that transmitters are aware of the network topology as well as the actual channel gains within a local neighborhood; the concept of normalized capacity was introduced by characterizing the maximum achievable rate using this channel knowledge as a fraction of the achievable rate using global channel state information. In \cite{jafar}, the authors considered a more restrictive scenario, called topological interference management (TIM), in which the transmitters are only aware of the topology with no information about the channel coefficients. It has been shown  that if the channel gains in the network remain constant for a sufficiently large time, then topological interference management is closely connected to the classical index coding problem (see, e.g., \cite{birk,baryossef,rouayheb}), and via this connection, a class of linear interference management schemes has been introduced which rely on the network topology knowledge to align the interference over time at unintended receivers.

In~\cite{topology,topology_isit}, the authors considered a class of retransmission-based schemes that only allows transmitters to resend their symbols in order to assist with the neutralization of interference at the receivers. Besides their simplicity, these schemes are robust to channel variations over time and were shown to be optimal in terms of the symmetric degrees-of-freedom (DoF) in many classes of topologies via the outer bounds developed in \cite{topology,topology_isit}. In this paper, we aim to characterize the ``best'' topologies for topological interference management with just retransmission; i.e., the topologies for which half symmetric DoF is achievable by a retransmission-based scheme. It is easy to see that in topologies where at least one transmitter is interfering at an undesired receiver, the symmetric DoF is upper bounded by $\frac{1}{2}$. Thus, the topologies in which half symmetric DoF is achievable represent the best topologies that one can hope for (from the degrees-of-freedom perspective).

When such retransmission restriction is not imposed on the TIM schemes and the channel gains in the network are assumed to remain constant for a long-enough period of time, the necessary and sufficient condition on the network topology for achieving half symmetric DoF was characterized in \cite{jafar}. For the case of retransmission-based schemes, a sufficient condition for the achievability of half symmetric DoF was derived in \cite{multialign}. However, the characterization of a necessary and sufficient condition has remained open. Moreover, it is not known whether or not extending the transmission schemes to any general linear scheme makes the half symmetric DoF achievable in a larger class of topologies.

In this work, we close the gap in the results on half symmetric DoF by introducing a necessary and sufficient condition under which the symmetric DoF of $\frac{1}{2}$ is achievable under retransmission-based schemes in the case of time-varying channels (i.e., without requiring the channels to remain fixed for a long enough time). The condition is based on a new concept of \emph{reduced conflict graph} which is a variation of the regular conflict graph used in previous works. Basically, the condition states that for a given topology, half symmetric DoF is feasible (under retransmission-based schemes) if and only if the corresponding reduced conflict graph is bipartite. Quite surprisingly, we also show that the same result holds even if we allow for all linear coding schemes. In other words, we show that for any given network topology with time-varying channels, if we cannot achieve half symmetric DoF by retransmission-based schemes, then there does not exist any linear scheme that can do so.

In order to prove the result, we show that when the condition is satisfied, there always exists a set of carefully-chosen transmitters in the network, such that by retransmission of their symbols at an appropriate time slot, we can neutralize all the interference at the receivers. More specifically, we show that only the transmitters which correspond to vertices in the reduced conflict graph that have no outgoing edges need to retransmit their symbols to allow for neutralization of interference at all receivers. Furthermore, through a novel upper bounding technique, we show that if the condition on the reduced conflict graph is not satisfied, no linear scheme can achieve half symmetric DoF, hence proving the necessity of the condition. The upper bound is based on the fact that if half symmetric DoF is achievable in a topology, then any pair of transmitters which are causing interference at a third receiver should both be silent for half the time slots; i.e., their transmission should be ``physically'' aligned.

We also numerically analyze the half symmetric DoF feasibility condition under a network setting that models cell edge users in heterogeneous networks. We first evaluate the fraction of network realizations in which this condition is satisfied, and show that this fraction is significant for a reasonable range of the number of users in the network. Furthermore, we compare the distribution of the actual user rates achieved by our scheme with half the interference-free point-to-point rates and observe that the gap between them is significantly low, hence demonstrating that the considered retransmission-based coding scheme performs extremely well in neutralizing the interfering signals over time and its promising performance generalizes beyond the DoF analysis. We also illustrate that the considered scheme outperforms interference avoidance for representative example topologies. It is important to point out that another implication of the numerical results is that even if we account for the interference from all the transmitters (as opposed to the \emph{network topology} where the weak interfering links are entirely removed from the network), the coding scheme still provides gains over the benchmark schemes.

\section{System Model}\label{sec:model}
In this paper, we study $K$-user interference networks composed of $K$ transmitters $\{\text{T}_i\}_{i=1}^K$ and $K$ receivers $\{\text{D}_i\}_{i=1}^K$. Each transmitter $\text{T}_i$ intends to deliver a message $W_i\in\mathcal{W}_i$ to its corresponding receiver $\text{D}_i$. We assume that each receiver is subject to interference only from a specific subset of the other transmitters and the interference power that it receives from the rest of them is below noise level. This leads to a \emph{network topology} indicating the interference pattern in the network. 

We assume that transmitters have no knowledge about the realization of the channel gains except for the topology of the network. However, the receivers have full channel state information. We refer to this assumption as no CSIT (channel state information at the transmitters) beyond topology. Each transmitter $\text{T}_i$ encodes its message $W_i$ to a vector $\mathbf{x}_i^n$ of length $n$ and transmits this vector over $n$ time slots. Each receiver $\text{D}_j$ receives the following signal over $n$ time slots
\begin{align}\label{signal}
\mathbf{y}_j^n=H_{jj}^n \mathbf{x}_j^n + \sum_{i\in \text{IN}_j} H_{ij}^n \mathbf{x}_i^n + \mathbf{z}_j^n,
\end{align} 
where for each $i,j\in\{1,...,K\}$ such that $\text{T}_i$ and $\text{D}_j$ are connected together in the network topology, $H_{ij}^n$ is an $n\times n$ diagonal matrix with the $k^{\textrm{th}}$ diagonal element being equal to the value of the channel coefficient  between $\text{T}_i$ and $\text{D}_j$ in time slot $k$, $\text{IN}_j$ denotes the set of transmitters causing interference at $\text{D}_j$ and $\mathbf{z}_j^n$ denotes the noise vector at $\text{D}_j$ where each of its elements is an i.i.d. $\mathcal{CN}(0,N)$ random variable, $N$ being the noise variance. There is a power constraint of $\frac{1}{n}\mathbb{E}\left[\|\mathbf{x}_i^n\|^2\right]\leq P$ on the transmit signals. The non-zero channel gains $\left\{\text{diag}(H_{ij}^n): \forall j, \forall i\in \text{IN}_j\cup\{j\}\right\}$ are assumed to be drawn from an  arbitrary joint continuous distribution. This assumption is what we will refer to as ``time-varying'' channels.\footnote{The case where the channel gains remain constant for a long period of time does not fall into this category since the corresponding joint channel distribution is not continuous.} Upon receiving this signal, each receiver $\text{D}_j$ will generate an estimate of its desired message, called $\hat{W}_j$, based on its knowledge of the channel gains. The probability of error $\mathbb{P}_e$ is defined as the probability that at least one receiver makes an error in decoding its message; i.e.,
\begin{align*}
\mathbb{P}_e=\text{Pr}\left[\bigvee_{i=1}^K (\hat{W}_i\neq W_i)\right].
\end{align*}

The rate tuple $(R_1,...,R_K)$ where $R_i=\frac{\log |\mathcal{W}_i|}{n}, i\in\{1,...,K\}$ is said to be achievable if there exists an encoding and decoding scheme for which the probability of error $\mathbb{P}_e$ goes to zero as the blocklength $n\rightarrow\infty$. The DoF tuple $(d_1,...,d_K)$ is said to be achievable if there exists an achievable rate tuple $(R_1,...,R_K)$ such that $d_i=\lim_{P\rightarrow\infty} \frac{R_i}{\log P},\forall i\in\{1,...,K\}$. The symmetric DoF $d_{sym}$ is defined as the supremum $d$ for which the DoF tuple $(d,...,d)$ is achievable.
Next we will describe the classes of linear and retransmission-based coding schemes.

\subsection{Description of Linear Schemes}
Assume that each transmitter $\text{T}_i$ intends to send a vector $\mathbf{w}_i\in\mathbb{C}^{m_i}$ of $m_i$ complex symbols to its desired receiver $\text{D}_i$. This message is encoded to the transmit vector $\mathbf{x}_i^n=V_i^n \mathbf{w}_i$, where $V_i^n$ denotes the $n \times m_i$ linear \emph{precoding} matrix of transmitter $i$, which can only depend on the knowledge of topology. Under such a scheme, the received signal of receiver $j$ over the $n$ time slots in (\ref{signal}) can be rewritten as
\begin{align*}
\mathbf{y}_j^n=(H_{jj}^n V_j^n) \mathbf{w}_j+\sum_{i\in \text{IN}_j} (H_{ij}^n V_i^n) \mathbf{w}_i+\mathbf{z}_j^n.
\end{align*}

At receiver $j$, the interference subspace denoted by $\mathcal{I}_j$ can be written as
\vspace*{-0.03in}
\begin{align}
\mathcal{I}_j=\bigcup_{i\in \text{IN}_j} \mathsf{colspan}(H_{ij}^n V_i^n).
\end{align}
In order to decode its desired symbols, receiver $j$ projects its received signal subspace given by $\mathsf{colspan} (H_{jj}^n V_j^n)$ onto the subspace orthogonal to $\mathcal{I}_j$, and its successful decoding condition can be expressed as
\begin{align}\label{eq:decodability}
\mathsf{dim}\left(\mathsf{Proj}_{{\cal I}_{j}^c} \mathsf{colspan} \left(H_{jj}^n V_j^n\right)\right)\geq m_j.
\end{align}

If the above decodability condition is satisfied at all the receivers $\{\text{D}_j\}_{j=1}^K$, then the DoF tuple $\left(\frac{m_1}{n},...,\frac{m_K}{n}\right)$ is achievable under the aforementioned linear scheme.

\subsection{Description of Retransmission-Based Schemes}
Retransmission-based schemes are linear schemes where each transmitter $i$ intends to send only one symbol over $n$ time slots, i.e., $m_i=1$ and $V_i^n$ is an $n \times 1$ vector. In each time slot $k$, the transmitter is either inactive (the $k^{th}$ element of $V_i^n$ is $0$)  or (re-)transmits the same symbol (the $k^{th}$ element of $V_i^n$ is $1$). Therefore, the $n \times 1$ precoding vector $V_i^n$ is binary. If all the transmitted symbols are successfully decoded, then $\frac{1}{n}$ DoF is achieved for each user. 
\section{Main Result}\label{sec:result}

In this section, we state our main result in Theorem~\ref{thm:symdof} and prove it. We start by making the following definition of \emph{reduced conflict graphs}.

\begin{defn}\label{redgraph}
The reduced conflict graph of a $K$-user interference network is a directed graph $G=(V,A)$ with $V=\{1,2,...,K\}$. As for the edges, vertex $i$ is connected to vertex $j$ (i.e., $(i,j)\in A$) if and only if $i \neq j$ and the following two conditions hold:
\begin{enumerate}
\item Transmitter $i$ is connected to receiver $j$.
\item $\exists s,k \in \{1,2,\ldots,K\}\backslash\{i\}, s\neq k$ such that both transmitter $i$ and transmitter $s$ are connected to receiver $k$. 
\end{enumerate}
\end{defn}

Clearly, the difference between the reduced conflict graph defined above and the regular conflict graph (see, e.g., \cite[Ch. 2.2]{conflict}) is the additional condition 2 in Definition \ref{redgraph}, which suggests that the reduced conflict graph only accounts for the interference caused by the transmitters which are accompanied by at least one other interfering transmitter at some receiver node in the network. This implies that the set of edges in the reduced conflict graph is a subset of the edges in the regular conflict graph, hence the name ``reduced'' conflict graph.

Having the above definition, we state our main result.
\begin{thm}\label{thm:symdof}
For a $K$-user interference network with arbitrary topology and time-varying channels, half symmetric DoF can be achieved through retransmission-based schemes over a finite symbol extension if the reduced conflict graph of the network is bipartite. Furthermore, any linear scheme cannot be used to achieve half symmetric DoF over a finite symbol extension if the reduced conflict graph is not bipartite.
\end{thm}

We now state the following remarks to highlight some implications of the result.
\begin{remk}\label{remk:besttopologies}
As we show later in Section~\ref{sec:analysis}, the condition of Theorem~\ref{thm:symdof} for achieving half symmetric DoF with retransmission occurs frequently in scenarios of practical interest. 
\end{remk}
\begin{remk}\label{remk:avoidance}
It is easy to see that the interference avoidance scheme, which schedules the users in an independent set of the regular conflict graph of the network, can achieve half symmetric DoF if and only if the regular conflict graph is bipartite. Since the reduced conflict graph is the same as the regular conflict graph with the removal of edges that do not satisfy condition $2$ in Definition~\ref{redgraph}, there are examples where the reduced conflict graph is bipartite while the regular conflict graph is not. For these examples, half symmetric DoF is achievable by a retransmission-based scheme, while it cannot be achieved by interference avoidance. We show in Section~\ref{sec:analysis} that there are many such examples in practical settings.
\end{remk}

We now show a representative example where the achieved symmetric DoF by retransmission cannot be achieved by interference avoidance.

\begin{ex}
A 6-user topology for which half symmetric DoF is achievable by a retransmission-based scheme, but it is not achievable by interference avoidance is depicted in Figure \ref{ex_gain}.
\vspace*{-.28in}
\begin{figure}[h]
\centering
\begin{subfigure}{0.17\textwidth}
\centering
\includegraphics[trim = 2.68in 2.3in 5.4in 2.2in, clip,width=0.78\textwidth]{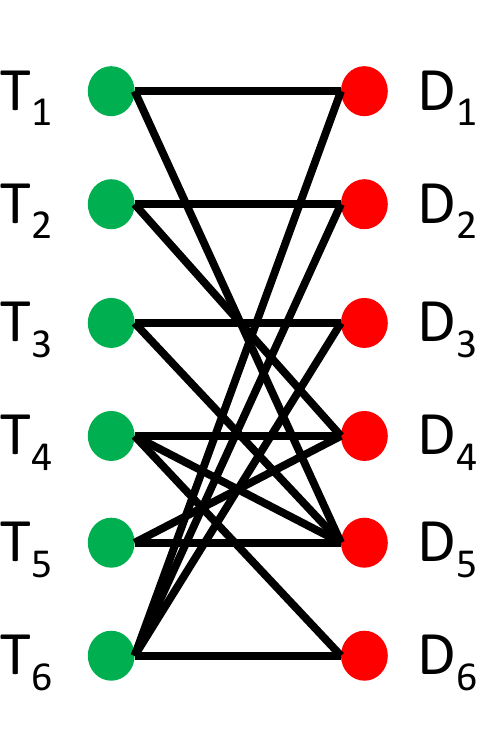}
\caption{}
\label{ex_gain}
\end{subfigure}
~
\begin{subfigure}{0.17\textwidth}
\centering
\includegraphics[trim = 5.4in 2.0in 1.9in 1.8in, clip,width=0.87\textwidth]{ex_top_avoidance}
\caption{}
\label{conf_graph}
\end{subfigure}
\caption{(a) A topology in which half symmteric DoF is achievable, but not by interference avoidance, and (b) the corresponding regular and reduced conflict graphs.}
\end{figure}

\noindent Figure \ref{conf_graph} shows the regular and reduced conflict graph of this network. The solid black edges are the edges in the reduced conflict graph and the dashed red edges are the ones that exist in the regular conflict graph but are removed in the reduced conflict graph. It is clear that the reduced conflict graph is bipartite (with the two partite sets being $\{1,3,4\}$ and $\{2,5,6\}$), hence $d_{sym}=\frac{1}{2}$ is achievable in this topology. However the addition of the dashed red edges in the regular conflict graph removes the bipartiteness property of the graph and therefore  interference avoidance cannot achieve half symmetric DoF in this topology. In fact, since the chromatic number of the regular conflict graph in Figure \ref{conf_graph} is 3, interference avoidance can only achieve a symmetric DoF of $\frac{1}{3}$.
\end{ex}

\begin{remk}
In \cite{jafar}, the scenario where the channel gains remain constant for a sufficiently long period of time was considered and linear schemes based on aligning the interference at unintended receivers were studied. It was shown that half symmetric DoF is achievable if and only if there is no conflict between any two nodes that cause interference at a third receiver, i.e., there is no internal conflict within an alignment set. It is easy to show that the condition in Theorem \ref{thm:symdof} implies the absence of internal conflicts but the opposite is not true; this is consistent with intuition as we make no assumption on the coherence time of the channel. Figure \ref{cex} shows an example topology in which there is no internal conflict, but the condition of Theorem \ref{thm:symdof} is not satisfied.
\begin{figure}[t]
\centering
\includegraphics[trim = 1.48in 2.3in 6.3in 2.5in, clip,width=.15\textwidth]{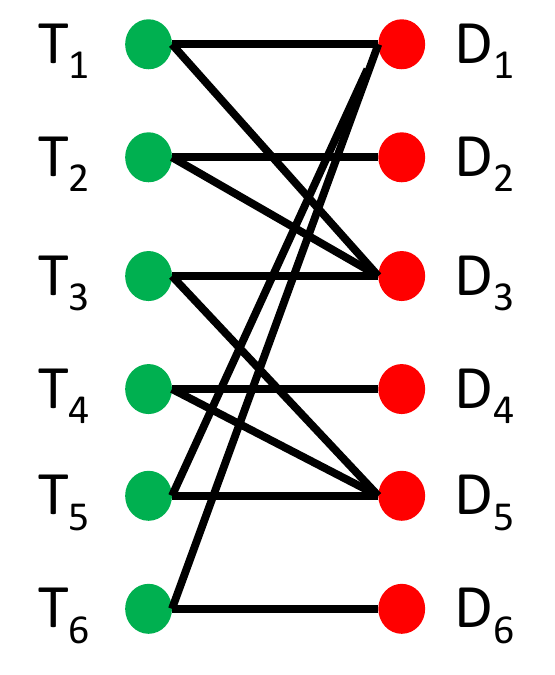}
\caption{A topology in which there is no internal conflict, but the reduced conflict graph is not bipartite.}
\label{cex}
\end{figure}
Interestingly, we show in Section~\ref{sec:analysis} that such topologies for which relying on channels remaining constant over time and using linear schemes to align the interference at unintended receivers can achieve half symmetric DoF while retransmission cannot achieve it comprise a negligible fraction of possible topologies in a heterogeneous network scenario of practical interest.
\end{remk}

Theorem~\ref{thm:symdof} follows from the converse and achievability proofs provided in Sections~\ref{sec:ub} and~\ref{sec:lb}, respectively.
\subsection{Converse}\label{sec:ub}
In the first part of the proof, we show that if the reduced conflict graph $G$ is not bipartite, then no linear coding scheme can be used to achieve half DoF for each user by coding over $n$ time slots. Without loss of generality, we assume that $n$ is even and prove that $\frac{1}{2}$ DoF cannot be achieved simultaneously for all users. We first prove the following lemma.
\begin{lem}\label{lem:ub}
For any two transmitters $i,j \in \{1,\ldots,K\}, i \neq j$, if $\exists k\in\{1,\ldots,K\}, k\notin\{i,j\}$ such that both transmitter $i$ and transmitter $j$ are connected to receiver $k$ and $\frac{1}{2}$ DoF is achieved for each of the messages $i,j$ and $k$, then there exist $\frac{n}{2}$ time slots for which both transmitter $i$ and transmitter $j$ are inactive. More precisely, $[V_i^n ~V_j^n]$ has $\frac{n}{2}$ zero rows.
\end{lem}
\begin{proof}
Since $\frac{1}{2}$ DoF is achieved for message $i$, then it follows from the decodability condition in~\eqref{eq:decodability} that rank$\left(H_{ii}^n V_i^n\right) \geq \frac{n}{2}$ for almost all realizations of $H_{ii}^n$. Because of the no-CSIT assumption and the fact that the channel coefficients are drawn from a joint continuous distribution, it also follows that rank$\left(H_{ik}^n V_i^n\right) \geq \frac{n}{2}$ for almost all realizations of $H_{ik}^n$. Similarly, since $\frac{1}{2}$ DoF is achieved for message $j$, it follows that rank$\left(H_{jk}^n V_j^n\right) \geq \frac{n}{2}$ almost surely. Now, if $\frac{1}{2}$ DoF is achieved for message $k$, then it follows from~\eqref{eq:decodability} that rank$\left(\left[H_{ik}^n V_i^n  \quad H_{jk}^n V_j^n\right]\right) \leq \frac{n}{2}$, as otherwise the complement of the interference subspace will have a dimension that is strictly less than $\frac{n}{2}$. 

Consider the case where both $V_i^n$ and $V_j^n$ have $\frac{n}{2}$ all-zero rows and the set of indices for the all-zero rows of $V_i^n$ differs from that of $V_j^n$; i.e., $\left[V_i^n \quad V_j^n\right]$ has less than $\frac{n}{2}$ zero rows. In this case, the matrix $H_{ik}^n V_i^n$ has $\frac{n}{2}$ linearly independent rows and the matrix $H_{jk}^n V_j^n$ has $\frac{n}{2}$ linearly independent rows almost surely, where the indices of the two identified sets of linearly independent rows are not identical. It follows that the matrix $\left[H_{ik}^n V_i^n  \quad H_{jk}^n V_j^n\right]$ has more than $\frac{n}{2}$ linearly independent rows almost surely, thereby contradicting the rank assumption. It hence suffices to assume in the rest of the proof that at least one of $V_i^n$ and $V_j^n$ has less than $\frac{n}{2}$ all-zero rows. 

Assume without loss of generality that $V_i^n$ has less than $\frac{n}{2}$ all-zero rows. It follows that almost surely there exists a vector $\mathbf{v}^n \in \mathsf{colspan}(H_{ik}^n V_i^n)$ which has more than $\frac{n}{2}$ non-zero entries. Since rank$\left(H_{ik}^n V_i^n\right) \geq \frac{n}{2}$, rank$\left(H_{jk}^n V_j^n\right)\geq \frac{n}{2}$ and rank$\left(\left[H_{ik}^n V_i^n  \quad H_{jk}^n V_j^n\right]\right) \leq \frac{n}{2}$, it follows that $\mathsf{colspan}(H_{ik}^n V_i^n)$ is an $\frac{n}{2}$-dimensional space that is identical to  $\mathsf{colspan}(H_{jk}^n V_j^n)$, and hence, $\mathbf{v}^n \in$ $\mathsf{colspan}(H_{jk}^n V_j^n)$ as well. Since $H_{jk}^n$ is drawn from a continuous distribution and the design of the precoding matrix $V_j^n$ cannot depend on the channel realization and $\mathbf{v}^n$ has more than $\frac{n}{2}$ non-zero entries, the probability that $\mathbf{v}^n$ lies in a specific $\frac{n}{2}$ dimensional subspace is zero, which is a contradiction.
\end{proof}

We use the result of Lemma~\ref{lem:ub} with the following lemma to prove the converse.
\begin{lem}\label{lem:conflict}
If transmitter $i$ is connected to receiver $k$, $k \neq i$ and $\frac{1}{2}$ DoF is achieved for each of message $i$ and message $k$ and each of $V_i^n$ and $V_k^n$ has $\frac{n}{2}$ zero rows, then $\left[V_i^n \quad V_k^n\right]$ has no zero rows.
\end{lem}
\begin{proof}
Assume otherwise. Since $\frac{1}{2}$ DoF is achieved for message $i$, it follows from the decodability condition of~\eqref{eq:decodability} that rank$\left(H_{ik}^n V_i^n\right)\geq \frac{n}{2}$ almost surely, hence the dimension of the projection of the received signal at receiver $k$ on the complement of the interference subspace is less than $\frac{n}{2}$, i.e., $\mathsf{dim}\left(\mathsf{Proj}_{{\cal I}_{k}^c} \mathsf{colspan} \left(H_{kk}^n V_k^n\right)\right) < \frac{n}{2}$. Thus, it follows from~\eqref{eq:decodability} that message $k$ cannot achieve $\frac{1}{2}$ DoF.
\end{proof}

Now, assume that $\frac{1}{2}$ DoF is achieved for each user in the network. We know from Lemma~\ref{lem:ub} that in the graph $G$, if vertices $i_1,i_2,\ldots,i_k$ have outgoing edges to a vertex $j$, then there are $\frac{n}{2}$ time slots for which transmitters $i_1,i_2,\ldots,i_k$ are silent. We then know from Lemma~\ref{lem:conflict} that if there is an edge from vertex $i$ to vertex $j$ and vertex $j$ has one or more outgoing edges, then in any of the $n$ time slots, exactly one of transmitter $i$ and transmitter $j$ is active. It follows that we can color all the vertices in $G$ that have one or more outgoing edges with two colors, such that any two colored neighboring vertices have different colors and all colored vertices that have outgoing edges to the same vertex have the same color. A valid $2-$coloring for $G$ is then complete by coloring each vertex with no outgoing edges by a different color from all its incoming neighbors. Since we found a valid $2-$coloring for $G$, it follows that $G$ is bipartite.
\vspace*{-.054in}
\subsection{Achievability}\label{sec:lb}
In the second part of the proof, we present a retransmission-based scheme that achieves half DoF for each user if the reduced conflict graph $G$ is bipartite. Consider a coding scheme over two time slots and suppose that each transmitter uses a point-to-point capacity achieving code and whenever it is activated in any of the two time slots, it transmits the selected codeword and otherwise it remains silent. The activation of transmitters is determined by the graph $G$ as follows. Let $P_1$ and $P_2$ be the two partite sets constituting $G$. If vertex $i$ has no outgoing edges, then transmitter $i$ is activated in both time slots. For any remaining user $j$ (corresponding to the nodes that have least one outgoing edge in $G$), if vertex $j$ is in $P_k$, $k\in\{1,2\}$, then transmitter $j$ is active in time slot $k$ and inactive in the other time slot. As an example, for the 6-user topology of Figure \ref{ex_gain}, the retransmission pattern can be written as follows,
\vspace*{-.08in}
\begin{align}\label{code}
\setlength{\arraycolsep}{4pt}
\begin{bmatrix}
V_1^2 & V_2^2 & V_3^2 & V_4^2 & V_5^2 & V_6^2 
\end{bmatrix}
=
\begin{bmatrix}
0 & 1 & 0 & 0 & 1 & 1\\
1 & 0 & 1 & 1 & 0 & 1
\end{bmatrix}.
\end{align}
Each column in~\eqref{code} corresponds to a user and each row corresponds to a time slot. For instance, transmitter 2 sends its codeword in time slot 1 and remains silent in time slot 2, whereas transmitter 6 repeats its codeword in both time slots.

We now show that successful decoding is possible for almost all realizations of the channel coefficients, and hence, $\frac{1}{2}$ DoF is achieved for each user over two time slots. For each receiver $i$ with more than two interfering links, transmitter $i$ is activated in a time slot where all interfering transmitters are silent, and hence, the successful decoding condition of~\eqref{eq:decodability} is guaranteed. In the example of Figure~\ref{ex_gain}, receivers $4$ and $5$ have two interfering links. Transmitters $2$ and $5$ are interfering at receiver $4$, and hence, in~\eqref{code}, transmitters $2$ and $5$ are silent in the second time slot where transmitter $4$ is active. Similarly, transmitters $3$ and $4$ are interfering at receiver $5$, and hence, both are silent in the first time slot where transmitter $5$ is active. For each receiver $i$ with one interfering link, it is either the case that transmitter $i$ is active in a time slot for which the interfering transmitter is silent (as is the case for receiver $6$ in the example), or it is the case that at least one of transmitter $i$ and the interfering transmitter is active in both time slots (as is the case for receivers $1$, $2$ and $3$ in the example); in both cases, the condition in~\eqref{eq:decodability} is satisfied.

\section{Network Analysis}\label{sec:analysis}
In this section, we numerically analyze the half symmetric DoF feasibility condition in a network setting that models cell edge users in heterogeneous networks. We consider a large square cell of side length $R$ in which $a\times b$ macro base stations (BS) are located uniformly on a grid, thereby dividing the whole area into $ab$ partitions. Moreover, there exist $K$ femto base stations located randomly in the entire cell area, with a  distribution that biases their location towards the partition boundaries. Each macro BS will have a corresponding receiver located randomly inside its partition, with a distribution that biases its location towards the center of the partition and each femto BS will have a small cell of radius $r$ around itself, inside which a receiver is located uniformly at random. A realization of such a network model is illustrated in Figure \ref{sim_network}.
\begin{figure}[h]
\centering
\includegraphics[trim = 2.35in .75in 1.9in 1.05in, clip,width=0.25\textwidth]{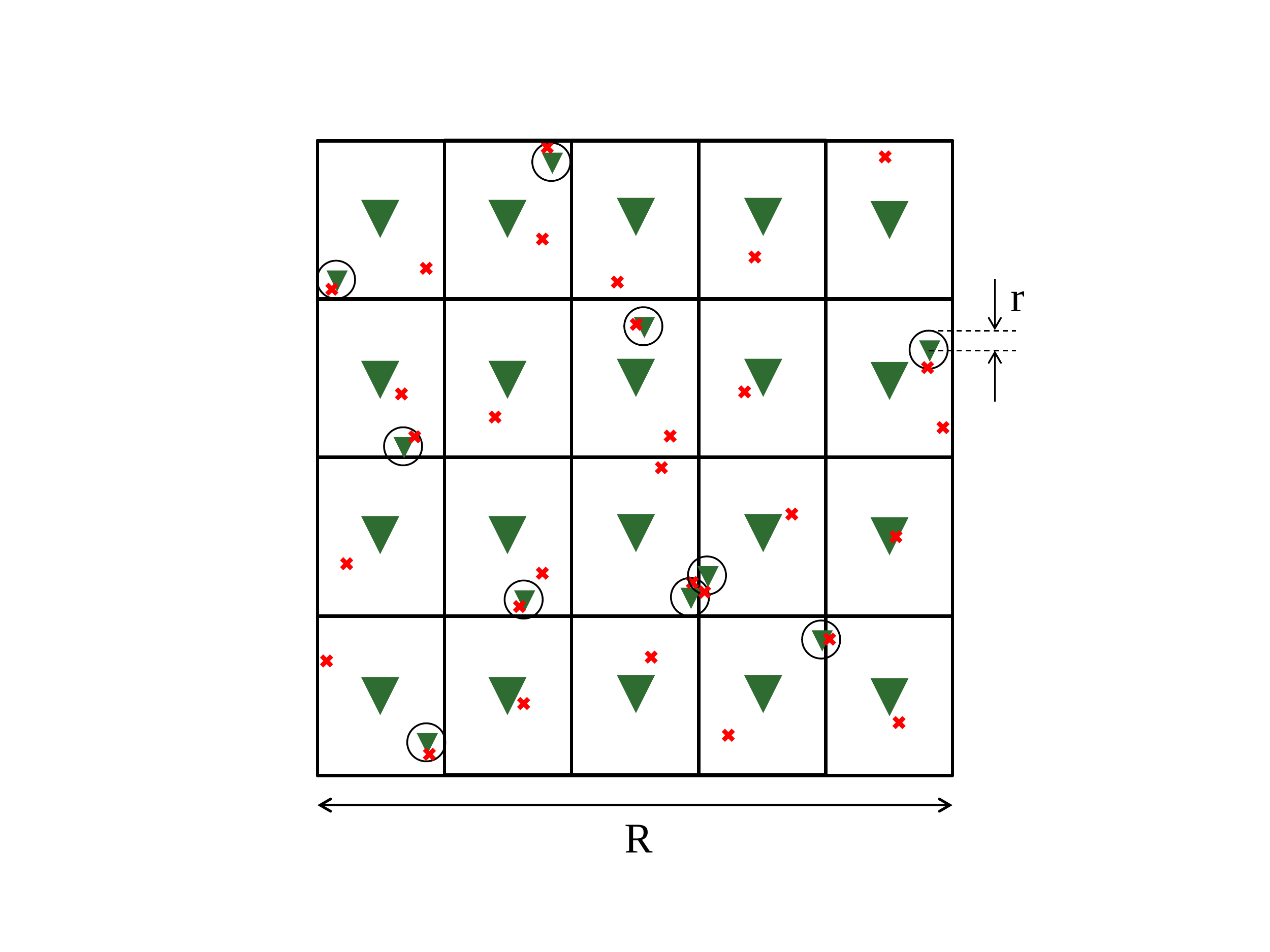}
\caption{The network model under study in which the green triangles and the red crosses are the transmitters and receivers, respectively.}
\label{sim_network}
\end{figure}


As for the parameters, the entire cell side length is taken to be $R=10km$, the femto BS small cell radius is taken to be $r=10m$, the transmit powers of the macro and femto base stations are taken to be 20 dBm and 10 dBm, respectively, and the noise variance is taken to be -100 dBm. We consider the ITU-1411 LoS channel model, in which the transmission loss (in dB) at distance $d$ equals
\begin{align*}
L=L_{bp}+6+\begin{cases}
20\log_{10} \left(\frac{d}{R_{bp}}\right) & \text{if } d\leq R_{bp}\\
40\log_{10} \left(\frac{d}{R_{bp}}\right) & \text{if } d> R_{bp}
\end{cases},
\end{align*}
where $R_{bp}=\frac{4 h_b h_m}{\lambda}$ denotes the breakpoint distance and $L_{bp}=\left|20\log_{10} \left(\frac{\lambda^2}{8\pi h_b h_m}\right)\right|$ denotes the transmission loss at $R_{bp}$, with $h_b$, $h_m$ and $\lambda$ being the BS antenna height, the mobile station antenna height and the transmission wavelength, respectively. The antenna heights are taken to be $1.5m$ and the carrier frequency is set to 2.4 GHz. We also consider shadowing with standard deviation of 10dB and Rayleigh fading with parameter 1 alongside the above channel model.

Each realization of the user locations and the channel gains results in a specific network topology, depending on whether the interfering signals are received above the noise level. We study the following questions:
$1)$ How often is the half symmetric DoF feasibility condition mentioned in Theorem \ref{thm:symdof} satisfied under the considered network setting?
$2)$ What is the actual rate that each user can achieve (beyond DoF analysis)?
$3)$ What are the rate gains of the proposed retransmission-based scheme beyond conventional schemes like interference avoidance?
$4)$ How often can the topological interference alignment scheme of~\cite{jafar} achieve half symmetric DoF and no retransmission-based scheme can achieve it?

Figure \ref{frac_half} shows the fraction of network realizations which satisfy the half symmetric DoF feasibility condition in Theorem \ref{thm:symdof} versus the number of femto base stations $K$.
\begin{figure}[h]
\centering
\includegraphics[trim = .5in 2.75in .59in 2.87in, clip,width=0.37\textwidth]{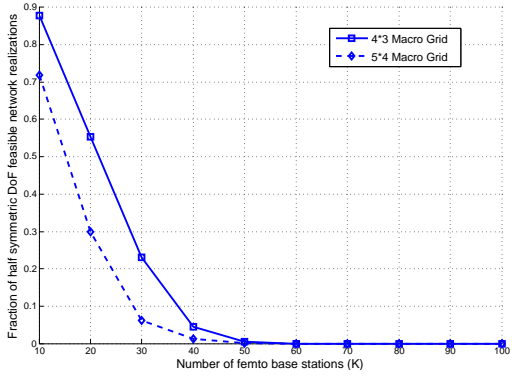}
\caption{Fraction of the network realizations which satisfy the half symmetric DoF feasibility condition.}
\label{frac_half}
\end{figure}
Obviously, increasing $K$ results in denser networks, and hence, reduces the probability that the half symmetric DoF feasibility condition in Theorem \ref{thm:symdof} is satisfied. However, for small values of $K$, we observe that a considerable fraction of the topologies satisfy this condition. For example, for the case of $K=10$, more than 87\% of the topologies satisfy the condition for the $4\times 3$ Macro BS grid and around 72\% of the topologies satisfy the condition for the $5\times 4$ Macro BS grid.

Focusing on the topologies in which the half symmetric DoF feasibility condition in Theorem \ref{thm:symdof} is satisfied, but interference avoidance cannot achieve it, we can evaluate the actual rate that each user can achieve and compare it with half of its point-to-point interference-free rate; i.e., $\frac{1}{2} \log (1+\text{SNR})$. Figure \ref{cdf} compares the CDF of the user rates for our coding scheme to the user rates under interference avoidance and $\frac{1}{2} \log (1+\text{SNR})$ in such topologies. The number of femto base stations is fixed at $K=20$ and the macro base stations are on a $4\times3$ grid.
\begin{figure}[h]
\centering
\includegraphics[trim = .85in 3.05in 1in 3.25in, clip,width=0.38\textwidth]{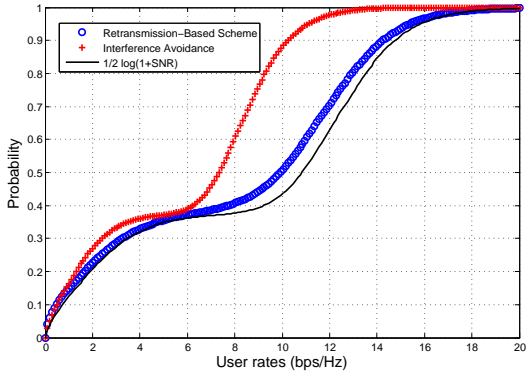}
\caption{Cumulative distribution function of the actual user rates in topologies for which the half DoF feasibility condition is satisfied, but interference avoidance cannot achieve it.}
\label{cdf}
\end{figure}
Interestingly, the distribution of the rates achieved under our simple retransmission-based scheme is very close to half of the interference-free rates. Moreover, there is a considerable gap between the performance of our scheme in comparison with interference avoidance, especially at high rates. This shows that in these topologies, our simple scheme which is only based on retransmission of the signals has a remarkable rate performance aside from its DoF optimality.

We also compare the performance of our retransmission-based scheme with the topological interference alignment scheme of \cite{jafar}, which takes advantage of channels remaining fixed over time and uses linear schemes to align interference over time at unintended receivers. For the $4 \times 3$ macro grid with $K=30$ femto BS's, $23\%$ of the generated topologies satisfy the half symmetric DoF feasibility condition using retransmission-based schemes and $23.1\%$ of topologies satisfy the half symmetric DoF feasibility condition for the alignment scheme of~\cite{jafar}. For the $5 \times 4$ macro grid with $K=30$, $6.1\%$ and $6.3\%$ of topologies satisfy the condition for the retransmission-based and alignment schemes, respectively. Interestingly, we observe that there are very few topologies in which assuming constant channels over time and using linear schemes and alignment of the interference can produce DoF gains over retransmission-based schemes and for most topologies, simpler retransmission-based schemes which are robust to channel variations suffice to achieve the optimal symmetric DoF.

In order to demonstrate the usefulness of our retransmission-based scheme in another cellular setting, we also evaluate our scheme under the hexagonal cell model depicted in Figure \ref{fig:hex_cells}.
\begin{figure}[h]
\centering
\includegraphics[trim = 1.9in .5in 2.45in .82in, clip,width=0.23\textwidth]{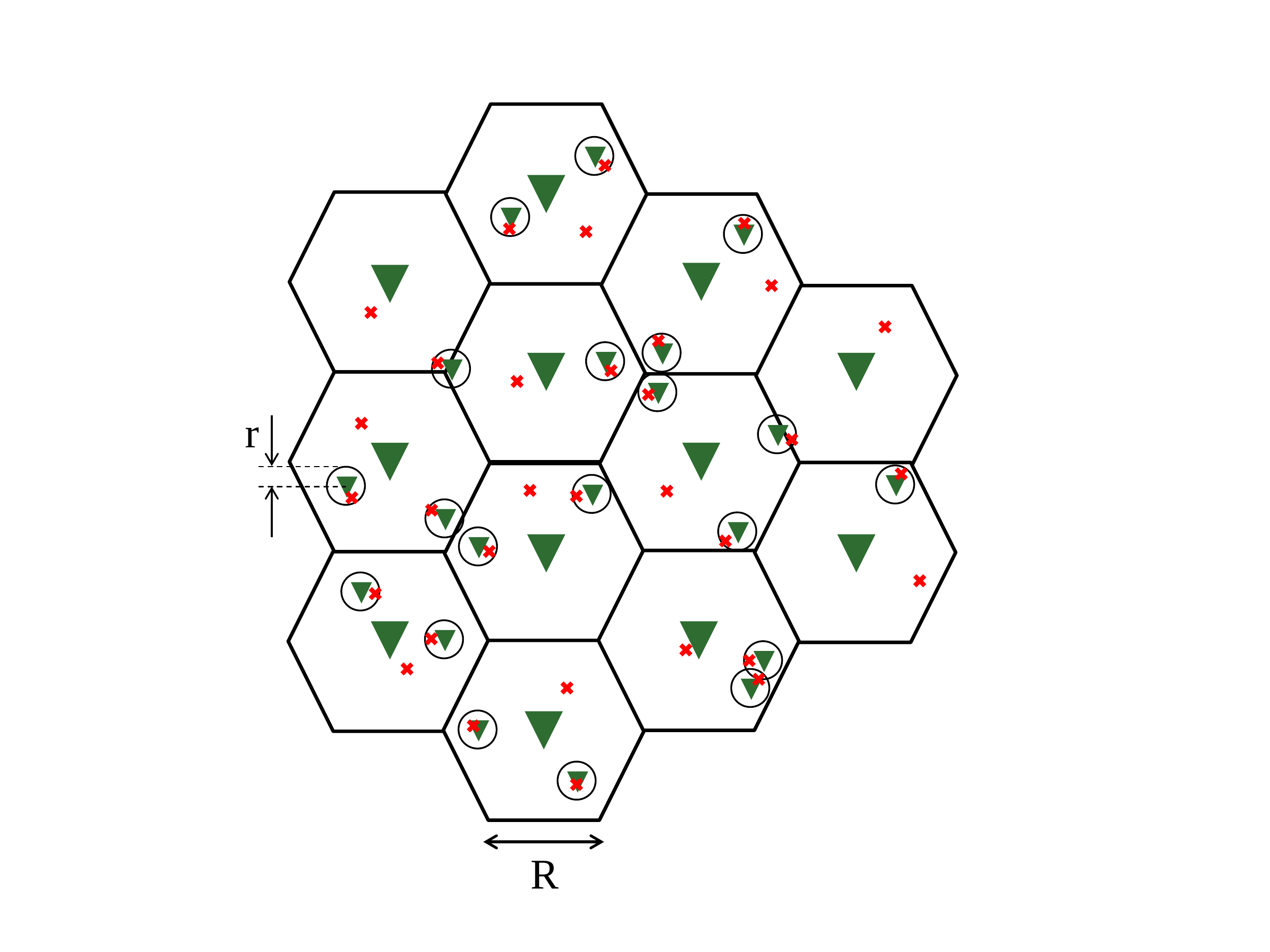}
\caption{Twelve hexagonal cells; the green triangles and the red crosses are the transmitters and receivers, respectively.}
\label{fig:hex_cells}
\end{figure}
This is a typical arrangement of twelve hexagonal cells, each with a side length of $R$. Similar to the previous setting, here there is a macro BS in the center of each cell, serving a receiver located randomly within the cell with a location distribution biased toward the center of the cell. Also, there exist $K$ femto BS's located randomly inside the cells, with locations biased toward the cell edges and each of them serves a corresponding user located uniformly within a distance $r$ of it.\footnote{We have modified the probability density function of the random distances inside a hexagon, derived in \cite{ref:hex}, in order to derive biased distributions for the locations of the macro mobile users and the femto base stations.}

We consider a cell side length of $R=1800m$ in this case in order to make the cell areas almost the same as the previous setting. The remaining parameters are completely kept the same as the previous case. For the sake of comparison, we consider the case of $K=20$ femto base stations and focus on topologies where our retransmission-based scheme can achieve half symmetric DoF while interference avoidance cannot achieve it. The distribution of the user rates is plotted in Figure \ref{fig:rates_hex_cells}.
\begin{figure}[h]
\centering
\includegraphics[trim = .3in 2.8in .3in 2.9in, clip,width=0.43\textwidth]{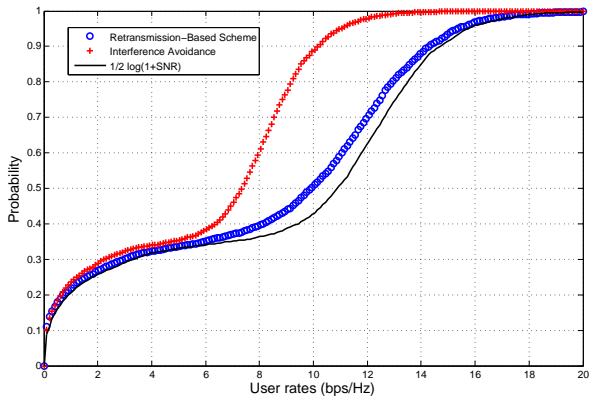}
\caption{Cumulative distribution function of the actual user rates for the case of hexagonal cells in topologies for which the half DoF feasibility condition is satisfied, but interference avoidance cannot achieve it.}
\label{fig:rates_hex_cells}
\end{figure}
We have also plotted the distribution of half of the point-to-point interference-free rates as a benchmark. As the figure depicts, the results look very similar to the case of rectangular cells, therefore suggesting that our retranmission-based scheme can effectively neutralize the interference in the case of hexagonal cells, as well.

\section{Concluding Remarks}\label{sec:conclusion}
In this paper, we focused on the problem of achieving half symmetric DoF for time-varying interference networks with no CSIT beyond topology information. We introduced a necessary and sufficient condition for the feasibility of half symmetric DoF under retransmission-based schemes. The achievability was proved through a careful design of the retransmission pattern, and the converse was shown for any linear scheme by restricting the silence pattern of the transmitters assuming that half symmetric DoF is achievable. We showed that under a particular configuration for heterogeneous networks, the characterized condition is satisfied in a considerable fraction of topologies. We also demonstrated that our simple retransmission-based scheme yields significant rate gains over interference avoidance in many network topologies and its gains are not restricted to high-SNR analysis.

There are numerous interesting future directions to consider following this work. An important direction would be to characterize the symmetric DoF for the topologies in which half symmetric DoF cannot be achieved through retransmission. An important progress has been made in this direction in \cite{ref:randomly_scaled_rows} in which the linear symmetric DoF has been fully characterized for a class of topologies, which includes the topology in Figure \ref{cex}. Second, the application of the retransmission-based schemes in a device-to-device (D2D) network setting can be investigated. Especially, it would be interesting to study how much improvement such schemes can provide alongside with the spectrum sharing mechanisms designed for D2D networks, such as the ITLinQ scheme proposed in \cite{itlinq,itlinq_isit,itlinq_dyspan} and the cached ITLinQ scheme introduced in \cite{cached_itlinq}. Third, the authors in \cite{sezgin} have considered time-varying topologies, in which the network topology is also changing over time. Extending the results of the paper to the case of time-varying topologies would hence be of particular interest. Finally, it is worthwhile to study the effect of cooperation between the nodes in the network (see, e.g., \cite{gesbert}) on the result by checking whether such cooperation can be useful in further extending the class of topologies for which half symmetric DoF is achievable.

\section*{Acknowledgement}
This work is supported by NSF Grants CAREER 1408639, CCF-1408755, NETS-1419632, EARS- 1411244, ONR award N000141310094, research grants from Intel and Verizon via the 5G project, and a gift from Qualcomm.

\bibliographystyle{ieeetran}
{\footnotesize
\bibliography{navid}}

\end{document}